\documentclass[a4paper,reqno,12pt]{article}
\usepackage{latexsym}
\usepackage{amsmath,amsthm,amssymb}

\def\t#1#2#3{#1\stackrel{#2}{\longrightarrow}{#3}}

\theoremstyle{plain}
\newtheorem{theorem}{Theorem}[section]
\newtheorem{lemma}[theorem]{Lemma}
\newtheorem{proposition}[theorem]{Proposition}

\theoremstyle{definition}

\newtheorem{example}[theorem]{Example}

\theoremstyle{remark}

\renewenvironment{proof}{\noindent{\it Proof}.}{\qed}

\title{Differential Transformations of Parabolic Second-Order Operators
in the Plane}
\author{S.P. Tsarev \thanks{The author was supported by the Russian
Foundation for Basic Research 06-01-00814.}
\\
Institute of Mathematics,\\
Siberian Federal University, \\
Svobodnyi avenue, 79 \\
660041, Krasnoyarsk. \\
e-mail: \texttt{sptsarev@mail.ru}
\and E. Shemyakova \thanks{The author was supported by the Austrian
Science Fund (FWF) under
project DIFFOP, Nr. P20336-N18.} \\
 Research Institute for Symbolic Computation, \\
 J. Kepler University,\\
 Altenberger Str. 69, \\
 Linz, Austria. \\
 e-mail: \texttt{kath@risc.uni-linz.ac.at}
}

\newcommand{\Sym}{\ensuremath \mathrm{Sym}}

\begin{document}
\maketitle

\begin{flushright}
    \emph{To Sergey Petrovich Novikov, as a development of one of his ideas.}
\end{flushright}

\bigskip

\section{Introduction}

The theory of transformations for hyperbolic second-order equations in the
plane, developed by Darboux, Laplace and Moutard, has many applications in
classical differential geometry~\cite{bianchi,Eis}, and beyond it in the theory
of integrable systems~\cite{fer-lap,NV}. These results, which were obtained
for the linear case, can be applied to non-linear Darboux-integrable
equations~\cite{2ZhS2001,anderson,forsyth,Gu}. In the last decade,
numerous generalizations of the classical theory have been developed.
Among them there are generalizations to the case of systems of
hyperbolic equations in the plane~\cite{Zh-St,St08,Tsarev99,ts05},
and generalizations to the case of
hyperbolic equations with more than two independent
variables~\cite{Athorne,ts06}. The non-hyperbolic case has been much less
investigated~\cite{LeRoux,Petren,pisati}.

Here, Darboux's classical results about
transformations with differential substitutions for hyperbolic equations
are extended to the case of parabolic equations.
Thus, consider for an arbitrary solution $u$ of the equation
\begin{equation} \label{op:L}
L u = 0, \quad L=D^2_{x} + a(x,y) D_x + b(x,y) D_y + c(x,y), \quad
b(x,y) \neq 0,
\end{equation}
some Linear Partial Differential Operator (LPDO) $M$ and a new
function $v(x,y)=Mu$. One can easily compute that in the generic
case $v$ satisfies an overdetermined system of linear differential
equations. However, there is some choice of $M$ which leads to only
one equation for $v$, namely, $L_1 v=0$, where $L_1$ is an operator
of the same form~(\ref{op:L}) allbeit with possibly different
coefficients $a_1(x,y)$, $c_1(x,y)$, $b_1\equiv b$. In this case we
say that we have a \emph{differential transformation} of operator
$L$ into operator $L_1$ with $M$, and denote this fact as
$\t{L}{M}{L_1}$. Also it is easy to notice that in this case there
must exist an operator $M_1$ such that the following equality holds:
\begin{equation} \label{main}
M_1 \circ L = L_1 \circ M \ ,
\end{equation}
that is the both parts of~(\ref{main}) define the left least
common multiple $lLCM(L,M)$ in the ring $K[D]=K[D_x,
D_y]$ of LPDOs in the plane.

For the case of hyperbolic operators of the form
\begin{equation} \label{op:L_hyperbolic}
L_H = D_{x}D_{y} + a(x,y) D_x + b(x,y) D_y + c(x,y)
\end{equation}
there are quite complete results on the possible
form of the operators $M$ that satisfy (\ref{main}) (see~\cite[Ch.
VIII]{Darboux2}): in the generic case the operator $M$ can be
determined (up to an arbitrary multiplier) from $M z_i=0$, $i=1,\ldots ,k$,
where $z_i(x,y)$ are independent solutions of $L_H z_i=0$. There are also
some degenerate cases. As was discovered by Darboux, one of those
degenerate cases is the classical Laplace transformation,
which is defined by the coefficients of operator~(\ref{op:L_hyperbolic})
only.
Relation~(\ref{main}) for the ``intertwining operator'' $M$ is widely
used in the study of integrability problems in two- and one-dimensional
cases~\cite{VSh93,BV00}.

In this paper, we prove general Theorem~\ref{prop:Dx+qDy+r:transf:exists}
that provides a way to determine transformations $\t{L}{M}{L_1}$ for
parabolic equations~(\ref{op:L}). It turned out (Theorem~\ref{th:decomp})
that transforming operators $M$ of some higher order
can be always represented
as a composition of some first-order operators that consecutively
define a series of transformations of the operators of the form~(\ref{op:L}).

Unlike the classical case of the Laplace and Moutard
transformations, the transformations considered in this paper are
not invertible. In this respect the problem in question is analogous
to the generic case that was considered in~\cite[Ch.
VIII]{Darboux2}) for operators~(\ref{op:L_hyperbolic}). As follows
from Theorems~\ref{prop:Dx+qDy+r:transf:exists},~\ref{th:decomp} for
parabolic operators~(\ref{op:L}) there are no degenerate cases like
Laplace transformations for \emph{arbitrary}
operators~(\ref{op:L_hyperbolic}): any differential transformation
of the operator (\ref{op:L}) can be determined by an operator $M$ of
the form~(\ref{fo:M1:det}). It is of interest to consider the
problem of the existence of an inverse
transformation~$\t{L_1}{N}{L}$. The order of the inverse may be
higher than the order of the initial transformation~$\t{L}{M}{L_1}$.
Examples show that the existence of such an inverse implies some
differential constrains on the coefficients of the initial
operator~$L$. In Sec.~\ref{sec:AthNimm} we show that these relations
can imply famous integrable equations, in particular, the Boussinesq
equation. This result is an analogue of
results~\cite{backes,fer-lap,Tz} for periodic chains of Laplace
transformations for the operators (\ref{op:L_hyperbolic}), which
also lead to integrable non-linear equations.

Authors are thankful to M.V. Pavlov for useful discussions.

\section{Basic Definitions and Auxiliary Results}
\label{sec:def}

Consider a field $K$ of characteristic zero with commuting derivations
$\partial_x, \partial_y$, and the ring of linear differential operators
$K[D]=K[D_x, D_y]$, where $D_x, D_y$ correspond to the derivations $\partial_x,
\partial_y$, respectively. In $K[D]$ the variables $D_x, D_y$ commute
with each other, but not with elements of $K$. For $a\in K$ we have
$D_i a = aD_i + \partial_i(a)$. Any operator $L \in K[D]$ has the form
$L = \sum_{i+j =0}^d a_{ij} D_x^i D_y^j$,
where $a_{ij} \in K$. The polynomial $\Sym_L =  \sum_{i+j
= d} a_{ij} X^i Y^j$ in formal variables $X, Y$ is called the (principal)
\emph{symbol} of $L$.

Below we assume that the field $K$ is differentially closed unless stated
otherwise, that is it contains solutions of (non-linear in the generic case)
differential equations with coefficients from $K$.

Let $K^*$ denote the set of invertible elements in $K$. For $L \in
K[D]$ and every $g \in K^*$ consider the gauge transformation $L
\rightarrow L^g = g^{-1} \circ L \circ g$. Then an algebraic
differential expression $I$ in the coefficients of $L$ is
(differential) \emph{invariant} under the gauge transformations (we
consider only these in the present paper) if it is unaltered by
these transformations. Trivial examples of invariants are the
coefficients of the symbol of an operator. A generating set of
invariants is a set using which all possible differential invariants
can be expressed.

\begin{theorem} \cite{Ibr02,movingframes}
The action of the gauge group on operators of
the form~\eqref{op:L} has the following generating system of invariants:
\begin{eqnarray*}
 I_1  &=& b \ , \\
 I_2 &=& c_{x}-a a_x/2
    -b a_y/2 -a_{xx}/2 \\
    && +(b_x a^2/4 - b_x c +b_x a_x/2 ) /b
  \ .
\end{eqnarray*}
\end{theorem}

Note that if an operator~\eqref{op:L} has only constant coefficients
then $I_1$ is a constant and $I_2=0$. If the field of coefficients $K$
contains quadratures (differentially closed), it is easy to prove
the inverse statement:

\begin{proposition} \label{prop:const_invar} Let the field of coefficients
$K$ be differentially closed. The equivalence class of~\eqref{op:L} with
respect to gauge transformations contains an operator with constant
coefficients if and only if $I_1$ is a constant and $I_2=0$.
\end{proposition}
\begin{proof}
Let $I_1=b$ have a constant value and $I_2=0$. Consider
an operator $L=D_{x}^2 + aD_x + b D_y + c$ from the equivalence class.
Using the gauge transformation with $g=\exp\big(- \frac{1}{2} \int a\, dx\big)$
one can make $a=0$. Then $I_2=0$ implies $0 = c_{x} - b_x c / b$.
Since $I_1=b$ is a constant, we have $c=c(y)$.
Applying the gauge transformation with
$g = e^{\int -c/b dy}$ to $L$  we obtain $L^g=D_{x}^2+bD_y$,
which has constant coefficients.
\end{proof}

So every operator~\eqref{op:L} with constant
$I_1=b$ and $I_2=0$ can be transformed into operator $D_{x}^2+D_y$
using substitution $y \mapsto \textrm{const}\cdot y$ and
gauge transformations.

\begin{lemma} \label{le:p=1}
Without loss of generality one can divide the symbols
$\Sym(M)=\Sym(M_1)$ by any non-zero $g \in K$. The operator $L$ and
the symbol of $L_1$ are left unchanged.
\end{lemma}
\begin{proof} Indeed, multiply the both sides of~\eqref{main}
by $1/g$ on the left: $\frac{1}{g} M \circ L = \frac{1}{g} L_1 g \circ \frac{1}{g}  M_1 =
 L_1^g  \circ \frac{1}{g}  M_1$.
Then ``new'' $M$ and $M_1$ have the coefficients of the ``old''
ones divided by $g$, while $L_1$ is subjected to the gauge transformation
with $g$, and, therefore,
its symbol is unchanged, while the other coefficients can be changed.
\end{proof}

\begin{lemma}[Simplification by gauge transformations]
\label{le:a=0}
In~\eqref{main} one can assume without loss of
generality that $a=0$, that is there exists a gauge transformation
that transforms $L$, $M$, $L_1$ and $M_1$ into operators of the same
form such that the coefficient of $L$ at $D_x$ is $0$, and the equality
$M \circ L = L_1 \circ M_1$~\eqref{main} is preserved.
\end{lemma}
\begin{proof} It is enough to apply the gauge
transformation with $g=\exp(- \frac{1}{2} \int a \, dx)$
to all operators in~\eqref{main}. This gauge transformation do not
alter the symbols of the operators, and, therefore, does not interfere
with the simplifications from Lemma~\ref{le:p=1}.
\end{proof}

\section{First-Order Transformations}
\label{sec:first_order_transfs}

Consider $L$ of the form~\eqref{op:L} and an operator $L_1$ of the same
form: $L_1 = D_{x}^2 + a_1(x,y) D_x + b_1(x,y) D_y + c_1(x,y)$. Then
a differential transformation of the first-order that transforms $L$ into $L_1$
exists if there exist
\begin{eqnarray*}
  M &=& p(x,y) D_{x} + q(x,y) D_y + r(x,y) \ ,\\
  M_1 &=& p_1(x,y) D_{x} + q_1(x,y) D_y + r_1(x,y)
\end{eqnarray*}
such that~(\ref{main}) holds. The comparison of the symbols implies
$p_1 = p$, $q_1=q$.

First consider the \emph{{case $\mathbf{p \neq 0}$, $\mathbf{q \neq
0}$ }}.

By lemma~\ref{le:p=1} without loss of generality one can assume $p=1$,
and $a=0$ by lemma~\ref{le:a=0}. Equating the coefficients
in~(\ref{main}) we have
  $a_1 = - 2 \frac{q_x}{q}$,
  $b_1 = b$,
  $c_1 = (-2 b q_x + b_x q + q^2 c + q^2 b_y
           + 2 q_x^2
           - b q_y q
           -q_{xx}q)/q^2$,
  $r_1 = r  - 2 (\ln q)_x$,
and two constrains on the coefficients of the operators $L$ and $M$:
$C_1 = 0$, $C_0 = 0$, where
\begin{eqnarray}
  C_0 &=& - 2 q c q_x
          + c_x q^2
          + q^3 c_y
          + 2 r b q_x
          - r b_x q
          - r q^2 b_y
          - 2 r q_x^2
          + \nonumber \\
      &&  + r b q_y q
          + r q_{xx} q
          + 2 q_x r_x q
          - b r_y q^2
          - r_{xx} q^2 \ ,
          \label{eq:C0} \\
  C_1 &=& -2 b q_x
          + b_x q
          + q^2 b_y
          + 2 q_x^2
          - b q_y q
          -q_{xx} q
          -2 q_x r q
          +2 r_x q^2 \ . \label{eq:C1}
\end{eqnarray}

We see from (\ref{eq:C0}), (\ref{eq:C1}) that
given the coefficients of the operator $L$, one can always find solutions
$r$, $q$ of these equations in the differentially closed field $K$,
that is every operator~(\ref{op:L}) admits
infinitely many transformations with different operators~$M$.
The equations (\ref{eq:C0}), (\ref{eq:C1}) for $r$, $q$ can be solved
explicitly with the help of two arbitrary (independent)
generic solutions of the equation~(\ref{op:L}). Indeed,
given a first-order operator $M$ that satisfies the constrain~(\ref{main}),
the following system of equations
\begin{equation}\label{SSU}
\left\{  \!\!   \begin{array}{l}
       Lu  =0, \\
       Mu  =0,
     \end{array}\right.
\end{equation}
is consistent and has a two-dimensional space of
solutions, which is parameterized, for example, by the
values $u(x_0,y_0)$, $u_y(x_0,y_0)$. In fact, we
can express the derivatives of $u$ of any order
with respect to $x$ in terms of its derivatives
with respect to $y$ from the second equation $Mu =0$.
Substituting those into the first equation $Lu =0$, we have an expression
for the second derivative $u_{yy}$, provided $q \neq 0$. On the other hand
the consistency of~(\ref{SSU}) is guaranteed by~(\ref{main}), which
can be rewritten as $q D_y Lu- D_x^2 Mu =0 \mod (L,M)$.
Conversely, a basis $z_1(x,y)$, $z_2(x,y)$ in the space of
solutions of~(\ref{SSU}) allows us to reconstruct $M$: the conditions $M z_1=0$,
$M z_2=0$ give a system of two linear algebraic equations for
the coefficients $r$, $q$, and we can easily determine the operator $M$:
\begin{equation} \label{M1ord}
Mu = \left|
  \begin{array}{ccc}
    u   & u_y & u_x \\
    z_1 & (z_1)_y & (z_1)_x \\
    z_2 & (z_2)_y & (z_2)_x \\
  \end{array}
\right|\cdot\left|
  \begin{array}{cc}
    z_1 & (z_1)_y \\
    z_2 & (z_2)_y \\
  \end{array}
\right|^{-1}.
\end{equation}
Since the values $z_i(x_0,y_0)$, $(z_i)_y(x_0,y_0)$ are lineally independent,
the denominator of this expression is non-zero.

Vice versa, the choice of two arbitrary lineally independent
solutions $z_1$, $z_2$ of the equation~(\ref{op:L}) defines
the operator $M$ by the formula~(\ref{M1ord}). The operator $M$ in its turn
implies a differential transformation of $L$, that is
the equality~(\ref{main}). Indeed,
compute the derivatives $v_x$,  $v_y$, $v_{xx}$ of the function $v=Mu$
for an arbitrary solution $u$ of the equation~(\ref{op:L}),
then using~(\ref{op:L}) we can remove all the terms that contain
$u_{xx}$, $u_{xxx}$, $u_{xxy}$. Using an
appropriate combination
$\tilde Lv=v_{xx}+a_1(x,y)v_{x}+b_1(x,y)v_{y}$ we can also remove
the terms with $u_{xy}$, $u_{yy}$, leaving $u_{x}$, $u_{y}$, $u$ only.
Since the expression~$\tilde Lv$ vanishes after
the substitution $u=z_i$ it must be proportional to $Mu$:
$\tilde Lv=\tilde L(Mu)=c_1(x,y)Mu$, which implies (\ref{main}) with
$L_1=\tilde L-c_1$ for an arbitrary function $u(x,y)$.

Note that in the considered case the coefficients at $D_y$, $D_x$ in
$M$ are non-zero. From now on we refer to such transformations as
$X+qY$-transformations. Below we consider the cases
when one or another of the coefficients is zero separately.
Therefore, we will prove the following statement:

\begin{theorem} \label{prop:Dx+qDy+r:transf:exists}
For every operator $L = D_{x}^2 + a D_y + b D_y + c$ there exist
infinitely many differential transformations
with operators $M =D_{x} + q(x,y) D_y + r(x,y)$. If $q\neq 0$
then the operator $M$ is defined by the conditions
$M z_1=0$, $M z_2=0$, where $z_i$ are two arbitrary chosen
independent solutions of the equation~(\ref{op:L}). The
operators of the form~$M = D_{x} + r(x,y)$ are defined by the choice of one
solution $z_1$ of the equation~(\ref{op:L}) and by the condition $M z_1=0$.
The intertwining operator of the form $M = D_{y} + r(x,y)$
does not exist for generic $L$.
\end{theorem}

The degenerate cases of operators $M$ of forms $M=D_x+r$ and $M=D_y+r$
are considered below.

 \emph{Case $\mathbf{p \neq 0}$, $\mathbf{q = 0}$
($M=D_x+r$)}

Without loss of generality one can assume $p=1$ and $a=0$.
If we equate the corresponding coefficients in~(\ref{main}), we have
  $a_1 = - \ln(b)_x$,
  $b_1 = b$,
  $c_1 = c+r \ln(b)_x -2 r_{x}$,
  $r_1 = r - \ln(b)_x$
and an equation
\begin{equation} \label{conds:Dx+r}
0 = - c \ln(b)_x
    + c_x
    - r^2 \ln(b)_x
    + 2 r r_x
    + r_x \ln(b)_x
    - b r_y
    - r_{xx} \ ,
\end{equation}
for $r$. We apply the same trick as in the non-degenerate case in order
to determine the operator $M$ in terms of solutions of
the initial equation~(\ref{op:L}). Now we choose \emph{one} solution
$z_1$ and require $M$ to satisfy the condition $M z_1=0$. We get
\begin{equation} \label{M1ordx}
M(u)= \left|
  \begin{array}{cc}
    u   & u_x \\
    z_1 & (z_1)_x \\
  \end{array}
\right|\cdot z_1^{-1} \ .
\end{equation}
Indeed, given an operator $M$ such that the intertwining equality~(\ref{main})
holds, an appropriate $z_1$ is found as a solution of the consistent
system~(\ref{SSU}), which now has a one-dimensional solution space.

Conversely, given a solution $z_1$ of the equation (\ref{op:L}),
$M$ can be found from~(\ref{M1ordx}), then for $v=Mu$ the
derivatives $v_x$, $v_y$, $v_{xx}$ are simplified using~(\ref{op:L}).

Then an appropriate combination
$\tilde Lv=v_{xx}+a_1(x,y)v_{x}+b_1(x,y)v_{y}$ contains only $u_{x}$ and $u$
(there are no terms with $u_{yy}$!). The obtained expression
$\tilde Lv$ vanishes if we substitute~$u=z_1$ and therefore it must be
proportional to $Mu$, which implies~(\ref{main}).

Later on we refer to such transformations as $X$-transformations.

\emph{Case $\mathbf{p = 0}$, $\mathbf{q \neq 0}$ ($M=D_y+r$)}

Without loss of generality we can assume $q=1$, $a=0$. If we equate
the corresponding coefficients in~(\ref{main}), we obtain in
particular  $r_x =  0$, $ c_y- r b_y- b r_y=  0$. Thus, $r=r(y)$ can
be found only for some particular functions $b, c$ and for an
arbitrarily chosen $L =D_{x}^2 + a D_x + b D_y + c$ there is no
differential transformations with $M=D_y + r$.

Notice also that an attempt to construct $M$ by the formula
$$M(u)= \left|
  \begin{array}{cc}
    u   & u_y \\
    z_1 & (z_1)_y \\
  \end{array}
\right|\cdot z_1^{-1}.
$$
would not lead to any success either: for such an operator
$M$ and $v=Mu$ the derivatives $v_x$, $v_y$, $v_{xx}$
simplified with~(\ref{op:L}) would contain
$u_{xy}$, $u_{yy}$, $u_{x}$, $u_{y}$,  $u$, and we cannot not
find an appropriate combination
$\tilde Lv=v_{xx}+a_1(x,y)v_{x}+b_1(x,y)v_{y}$ having only
$u_{y}$, $u$.

Therefore, Theorem~\ref{prop:Dx+qDy+r:transf:exists} is proved.

Note that when differential transformations with
$M=D_x+qD_y+r$ are applied to the operator~\eqref{op:L}, the
new values of the basic invariants
(that is the values of invariants $I_1$ and $I_2$ for $L_1$)
are
\begin{eqnarray*}
 I_1^1  &=& I_1 = b \ , \\
 I_2^1 &=& I_2 -2 b q_{xx}/q^2
                 - b_x^2/(qb)
                 - b_x b_y /b
                 + b_{xx}/q
                 + b_{xy}
                 - b_x q_x/q^2
                 + 4 q_x^2 b /q^3 \ .
\end{eqnarray*}

When differential transformations with
$M=D_x+r$ are applied the new values of the basic invariants are
\begin{eqnarray*}
 I_1^1  &=& I_1 = b \ , \\
 I_2^1 &=& I_2-1/4  (
                       8 b^3  r_{xx}
                       - 12  b_{x}  r_{x}  b^2
                       +8  r  b  b_{x}^2 \\
    &&                 -4  r  b^2  b_{xx}
                       +2  b^2  b_{y}  b_{x}
                       -9  b_{x}^3
                       - 2  b^2  b_{xxx}
                       - 10  b_{x}  b  b_{xx}-2  b^3  b_{xy})/b^3 \ .
\end{eqnarray*}

\begin{example} Consider an operator
\[
L= D_{xx} + \frac{2x+2y}{x^2} D_y - \frac{2}{x^2} \ .
\]
The equation $L(z)=0$ has the following solutions $z_1 = x^2$, $z_2 = x+y$.
Using the determinantal formula~\eqref{M1ord} compute
\[
M = D_x + \frac{x+2y}{x}  D_y -
 \frac{2}{x}  \ ,
\]
and $M_1 = D_x + \frac{x+2y}{x}  D_y- \frac{2}{x+2y}$,
\[
L_1 = D_{x}^2 - \frac{4y}{x(x+2y)} D_x + (\frac{2}{x}+
\frac{2y}{x^2})D_y - \frac{6}{(x+2y)x}- \frac{4y}{(x+2y)x^2} \ .
\]
Note that $L_1$ cannot be obtained from $L$ by any
gauge transformation. Indeed, the value of the invariant
$I_2$ for $L$ is $I_2 = \frac{2}{x^2(x+y)}$, while the value
of $I_2$ for $L_1$ is $I_2^1 =\frac{2(x^2-2xy-4y^2)}{x(x+y)(x+2y)^3}$.
\end{example}

\begin{example} \label{ex:cc2dxdy}
Applying the differential transformation with $M=D_x + q(x,y) D_y +
r(x,y)$ to $L=D_{x}^2+D_y$ (provided conditions~\eqref{eq:C0} and
\eqref{eq:C1} are satisfied or equivalently, provided $M$ is
in the form~(\ref{M1ord})) we have
$M_1=D_x + q(x,y) D_y + r - 2 (\ln q)_x$ and
\begin{eqnarray*}
L=D_{x}^2+D_y  &\qquad \longrightarrow  \qquad& L_1 = D_{x}^2
-2q_x/q D_x + D_y +
(q_y q+q_{xx} q - 2q_x^2+ 2q_x)/q^2 \ , \\
I_2=0 &\qquad \longrightarrow  \qquad& I_2^1= - 2q_{xx}/q^2
                +4q_x^2/q^3 \ .
\end{eqnarray*}
\end{example}

\begin{example} \label{ex:cc2dx}
Applying the differential transformation with $M=D_x + r(x,y)$ to
$L=D_{x}^2+D_y$ (provided the condition~\eqref{conds:Dx+r} is
satisfied or equivalently, provided $M$ is in the form(\ref{M1ordx}))
we have $M_1=M$ and
\begin{eqnarray*}
L=D_{x}^2+D_y  &\qquad \longrightarrow  \qquad&
L_1 = D_{x}^2 + D_y - 2 r_x \ , \\
I_2=0 &\qquad \longrightarrow  \qquad& I_2^1= - 2r_{xx}\ .
\end{eqnarray*}
\end{example}

\section{Transformations of Arbitrary Order}
\label{sec:darboux}

We show that differential transformations of arbitrary order
of a generic operator~\eqref{op:L} can be expressed
in terms of some number of partial solutions of~\eqref{op:L}.
In~\cite[Ch.VIII]{Darboux2} analogous formulae were introduced for
hyperbolic operators~\eqref{op:L_hyperbolic}.

First of all, given some transforming operator $M$ of
higher order
satisfying~(\ref{main}), we can use the operator $L$ to remove
all terms having derivatives with respect to
$y$ (generally speaking, this manipulation increases the order of $M$).
The resulting operator has the form
\begin{equation}\label{MMM}
     M= \sum_{i=0}^{m}q_{i}(x,y)D_x^i, \quad q_m \neq 0 \ .
\end{equation}
Below we call the corresponding transformation an
\emph{$(m)$-transformation}.

\begin{theorem} Given an operator~\eqref{op:L} and $m$
lineally independent generic partial solutions $z_1, \dots , z_{m}$
of the corresponding equation $L(z)=0$, then there exists a
differential transformation with
\begin{equation} \label{fo:M1:det}
Mu = \vartheta(x,y) \left|
  \begin{array}{cccc}
    u      & \frac{\partial u}{\partial x} & \dots &
    \frac{\partial^m u}{\partial x^m}  \\
    z_1      & \frac{\partial z_1}{\partial x} & \dots &
    \frac{\partial^m z_1}{\partial x^m}  \\
    \vdots & \vdots       & \vdots& \vdots  \\
    z_m      & \frac{\partial z_m}{\partial x} & \dots &
    \frac{\partial^m z_m}{\partial x^m}
  \end{array}
\right|
\end{equation}
where $\vartheta(x,y)$ is arbitrary. Conversely, every
$(m)$-transformation of an operator of the form~\eqref{op:L} corresponds
to some operator $M$ of the form~\eqref{fo:M1:det}.
\end{theorem}
\begin{proof} Having computed the derivatives $v_x$,  $v_y$, $v_{xx}$ of
$v=Mu$ for an arbitrary solution $u$ of equation~(\ref{op:L}), we
use~(\ref{op:L}) as above to remove all terms that contain
derivatives with respect to $y$. The remaining terms will contain
only some linear combinations of the derivatives
$D_x^s u$, $s=0,\ldots , m+2$.
Choosing some appropriate combination
$\tilde Lv=v_{xx}+a_1(x,y)v_{x}+b_1(x,y)v_{y}$ we can remove terms
with $D_x^{m+2} u$, $D_x^{m+1} u$, and leave terms with
$D_x^s u$,  $s=0,\ldots , m$ only. Since the resulting
expression $\tilde Lv$ vanishes when we substitute any
$u=z_i$, we conclude that it must be proportional to
$Mu$: $\tilde Lv=\tilde L(Mu)=c_1(x,y)Mu$,
which implies~(\ref{main}) with $L_1=\tilde L-c_1$ for
an arbitrary function $u(x,y)$. The only requirement is the non-vanishing
of the Wronskian $\det(D_x^j z_i)$, $i=1, \ldots , m$, $j=0,
\ldots , m-1$.

Conversely, given the intertwining operator $M$ of the form (\ref{MMM})
satisfying~(\ref{main}), consider the system~(\ref{SSU}). The consistency
of the system is equivalent to~(\ref{main}), which allows us
to choose a basis of its $m$ solutions with
non-vanishing
Wronskian $\det(D_x^j z_i)$, $i=1,
\ldots , m$, $j=0, \ldots , m-1$, and obtain the required
form~(\ref{fo:M1:det}) of the operator $M$.
\end{proof}

\begin{theorem}  \label{th:decomp}
An arbitrary $(m)$-transformation of an operator (\ref{op:L}) with $m>1$
can be represented as a composition of first-order differential
transformations.
\end{theorem}
\begin{proof} Consider an operator $M$ in the form~(\ref{fo:M1:det})
and the corresponding solutions $z_i$. Then $z_1$ generates a
first-order transformation with $\hat M$ of the form~(\ref{M1ordx}),
which transforms $L$ into some $\hat L$ of the same
form~(\ref{op:L}). Others $z_i$, $i=2, \ldots , m$ are transformed
into solutions $\hat z_i=\hat M z_i$ of the equation $\hat L \hat z
=0$. Since $Mz_1=0$, $\hat Mz_1=0$, then if we divide the ordinary
differential operator $M$ by $\hat M$, the remainder is zero: $M=P
\hat M$, $P \in K[D_x]$. (\ref{main}) implies that the operator
$lLCM(L, M)=M_1 L=L_1M$ is divisible by $lLCM(L,\hat M)=\hat M_1 L =
\hat L \hat M$, that is $M_1 L=N_1\hat M_1 L=N_1 \hat L \hat
M=L_1M=L_1P\hat M$, which implies $N_1 \hat L=L_1P$. Thus we have
obtained an intertwining operator $P$, whose order is less by one,
such that $\t{\hat L}{P}{L_1}$. The induction by the order $m$ of
the intertwining operator completes the proof.
\end{proof}

\section{Generalized Moutard Transformations
and Differential Transformations. Periodical Differential Transformations}
\label{sec:AthNimm}

An important subclass of the considered class of the parabolic operators
are operators
\begin{equation} \label{moutard}
L=D_{x}^2-D_y + c(x,y) \ .
\end{equation}

In~\cite{athorne_Nimmo_moutard}, a modification of Moutard transformations
for such operators was suggested and applications to the construction
of solutions in the Kadomtsev---Petviashvili (KP) hierarchy of equations
were given. As we show below, some of the examples considered
in~\cite{athorne_Nimmo_moutard} can also be obtained by our method.
Direct application of the above results proves the following lemma.

\begin{lemma} \label{le:moutard:X-tranf:basics}
 $X$-transformations preserve the class of the operators~\eqref{moutard}.
For $M=D_x + r(x,y)$ the condition~\eqref{conds:Dx+r} for the
existence of such transformations has the following form:
\begin{equation} \label{moutard:Dx+r:conds}
c_x+2 r r_x + r_y - r_{xx} = 0 \ ,
\end{equation}
and
\begin{equation} \label{moutard:Dx+r:L1}
M_1 = M \ , \quad L_1 = D_{x}^2 - D_y + c-2 r_{x} \ .
\end{equation}
The basic invariant $I_2$ transforms as follows:
$I_2= c_{x} \ \longrightarrow  \ I_2^1= c_x - 2 r_{xx}$.
If the operator $M$ is given in the form~(\ref{M1ordx}) for some
partial solution $z_1=z_1(x,y)$ of the equation $L(z)=0$, we have
\[
L_1 = D_{x}^2 - D_y - \frac{2z_{1x}^2-z_{1} z_{1y} -z_{1xx}
z_{1}}{z_{1}^2} \ .
\]
\end{lemma}

Note that $X+qY$-transformations do not preserve
the class of operators~\eqref{moutard}:

\begin{example} The equation $(D_{x}^2 - D_y)z=0$ has partial solutions $z_1 = x$, $z_2
= e^{x+y}$. The formula~(\ref{M1ord}) implies
$M = D_x + \left(\frac{1}{x}-1 \right) D_y - \frac{1}{x}$
and
$L_1 = D_{x}^2 - \frac{2}{x(x-1)} D_x -D_y - \frac{2}{x(x-1)}$.
However, the gauge transformation with $g-(x-1)/x$
reduces $L_1$ to the form~(\ref{moutard}):
$L_1^g=D_{x}^2 - D_y-\frac{2}{(x-1)^2}$.
\end{example}
This example and the one below show that classical examples of
functions $c(x,y)$ obtained in~\cite{athorne_Nimmo_moutard} can also
be obtained by the application of one or several differential
transformations. Actually, both approaches can be considered as
two-dimensional generalizations of Darboux transformations for the
one-dimensional Schr\"odinger operator $D_x^2 - c(x,y)$.

\begin{example} Consider a differential transformation of $L=D_{x}^2+D_y$
with $M=D_x + r(x,y)$. Choosing $r=1/2-\tanh(x+y)$ satisfying the
condition~(\ref{moutard:Dx+r:conds}) of the existence of the
transformation, we have
\begin{eqnarray*}
L=D_{x}^2+D_y  &\qquad \longrightarrow  \qquad&
L_1 = D_{x}^2 + D_y + \frac{2}{\cosh(x+y)^2} \ , \\
I_2=0 &\qquad \longrightarrow  \qquad& I_2^1= - \frac{ 4
\sinh(x+y)}{\cosh(x+y)^3} \ .
\end{eqnarray*}
\end{example}

Now we study the invertibility of a given transformation
$\t{L}{M}{L_1}$,
that is, the possibility of finding a transformation $\t{L_1}{N}{L}$,
possibly of higher order.
\begin{example} $X$-transformation of the
operator $L=D_{x}^2-D_y - x^4 + 2x$
with $M=D_x+ x^2$ results in the following operator: $L_1 = D_{x}^2-D_y -
x^4 - 2x$. This transformation has the inverse $X$-transformation
with $N=D_x- x^2$.
\end{example}

As the simplest examples show, an inverse transformation does not
exist for a generic operator $L$. In fact the existence of an
inverse transformation implies a system of constrains on the
coefficients of $L$. In some cases, it produces known integrable
equations. First, Theorem~\ref{th:decomp} implies that the existence
of an inverse transformation, that is the existence of a composition
$\t{L}{N\cdot M}{L}$, is equivalent to the existence of a
transformation $P=N\cdot M$ of higher order that transforms the
operator $L$ into itself: $P_1 \cdot L = L \cdot P$. For
operators~(\ref{moutard}) the existence of such an operator implies
a particular case of the standard problem of classification of Lax
pairs: for $P$ of order one or two of the form~(\ref{MMM}) this
leads to potentials $c(x,y)$ of simple form; the existence of an
operator~$P=p_3(x,y)D_x^3 + p_2(x,y)D_x^2 + p_1(x,y)D_x + p_0(x,y)$
of the third order implies $P_1=P$ and $P=4 D_x^3 + 6cD_x +
p_0(x,y)$ (up to some simple transformations) and the system
\begin{equation}\label{UB}
    \left\{\begin{array}{ll}
      (p_0)_x & = 3(c_y +c_{xx}),\\
      (p_0)_y & = 3c_{xy} -6cc_x -c_{xxx},
    \end{array}\right.
\end{equation}
that is the famous Boussinesq equation for $c$:
$$c_{yy} = -(c^2+c_{xx}/3)_{xx}.
$$
The system (\ref{UB}) coincides with the well-known representation
(\cite[formula (7)]{KN80}) for the Kadomtsev-Petviashvili equation
in the stationary case $U_t=0$, which gives the Boussinesq equation.

\bibliographystyle{plain}

\end{document}